\documentclass[conference,a4paper,10pt]{IEEEtran}
\usepackage{bbm,dsfont,mathrsfs,fixmath,amsthm}
\usepackage{amsmath,amssymb,graphicx} 
\usepackage{stmaryrd,cite}
\usepackage{amsmath,amstext,amsfonts,amssymb}
\usepackage{epsfig}
\usepackage{exscale}
\usepackage{enumerate}
\usepackage{mathtools}
\usepackage{multirow}

\usepackage{algorithm,algorithmic}

\usepackage{amssymb}
\usepackage{amscd}
\usepackage{amsmath}
\usepackage{bm}
\usepackage{epsfig}
\usepackage{graphics}
\usepackage{psfrag}
\usepackage{rotating}
\usepackage{amsmath}
\usepackage{amsfonts}
\usepackage{url}
\usepackage{color}
\usepackage{epstopdf}
\usepackage{amsthm}
\usepackage{tikz}
\usepackage{mathtools}
\usetikzlibrary{arrows,shapes,decorations,backgrounds}
\usepackage[final]{pdfpages}
\usepackage{multicol}
\usepackage{cite}
\usepackage{graphicx}
\usepackage{caption}
\usepackage{subfig}
\usepackage{breqn}

\newtheorem{theorem}{Theorem}

\newtheorem{remark}{Remark}

\newtheorem{proposition}[theorem]{Proposition}

\newcommand{\Bern}{\mathsf{Bern}}

\newfont{\bbb}{msbm10 scaled 500}

\newfont{\bb}{msbm10 scaled 1100}

\newcommand{\EE}{\mbox{\bb E}}

\newcommand{\Nc}{{\cal N}}

\newcommand{\Pc}{{\cal P}}

\newcommand{\Rc}{{\cal R}}

\newcommand{\Tc}{{\cal T}}

\newcommand{\Xc}{{\cal X}}
\newcommand{\Yc}{{\cal Y}}

\newcommand{\eqdef}{\stackrel{\Delta}{=}}

\newcommand{\snr}{\mathsf{snr}}

\newcommand{\mkv}{-\!\!\!\!\minuso\!\!\!\!-}

\include{macros}

\begin{document}

\title{On the Relevance-Complexity Region of Scalable Information Bottleneck}

\author{\IEEEauthorblockN{Mohammad Mahdi Mahvari}
\IEEEauthorblockA{
\thanks{}
Sharif University of Technology \\
Tehran, Iran\\
 {\tt mahdimahvar@gmail.com}
}
\and
\IEEEauthorblockN{Mari Kobayashi}
\IEEEauthorblockA{
\thanks{The work of M. Kobayashi was supported by DFG.}
Technical University of Munich \\
Munich, Germany\\
 {\tt mari.kobayashi@tum.de}
}
\and
\IEEEauthorblockN{Abdellatif Zaidi}
\IEEEauthorblockA{
\thanks{}
Universite Paris-Est \\
Champs-sur-Marne, 77454, France\\
 {\tt abdellatif.zaidi@u-pem.fr}
}
}

\maketitle
\begin{abstract}
The Information Bottleneck method is a learning technique that seeks a right balance between accuracy and generalization capability through a suitable tradeoff between compression complexity, measured by minimum description length, and distortion evaluated under logarithmic loss measure. 
In this paper, we study a variation of the problem, called scalable information bottleneck, where the encoder outputs multiple descriptions of the observation with increasingly richer features. The problem at hand is motivated by some application scenarios that require varying levels of accuracy depending on the allowed level of generalization. 
First, we establish explicit (analytic) characterizations of the relevance-complexity region for memoryless Gaussian sources and memoryless binary sources. Then, we derive a Blahut-Arimoto type algorithm that allows us to compute (an approximation of) the region for general discrete sources. 
Finally, an application example in the pattern classification problem is provided along with numerical results. 
\end{abstract}

\section{Introduction}
In statistical (supervised) learning models, one seeks to strike a right balance between accuracy and generalization capability. Many existing machine learning algorithms generally fail to do so; or else only at the expense of large amounts of algorithmic components and hyperparameters that are heavily tuned heuristically for a given task (and, even so, generally fall short of generalizing to other tasks). 
The Information Bottleneck (IB) method~\cite{tishby1999information}, which is essentially a remote source coding problem under logarithmic loss fidelity measure, 
attempts to do so by finding a suitable tradeoff between algorithm's robustness, measured by compression complexity, and 
accuracy or relevance, measured by the allowed average logarithmic loss (see e.g. ~\cite{zaidi2020information, goldfeld2020information}). More specifically, for a target (label) variable $X$ and an observed variable $Y$, the IB finds a description $U$ that is maximally informative about $X$ while being minimally informative about $Y$, where informativeness is measured via Shannon's mutual information.

We study a variation of the IB problem, called {\it scalable information bottleneck}, where the encoder outputs $T\geq 2$ descriptions with increasingly richer features. 
This is motivated primarily by application scenarios in which a varying level of accuracy is required depending on the allowed and/or required level of complexity. The model is illustrated in Fig.~\ref{fig:ScalableRSC} for the case of $T=2$. As an example, one may think about the simple scenario of making inference about a moving object on the road. A decision maker which receives only coarse information about the object would have to merely identify its type (e.g., car, bicycle, bus), while one that receives also refinement information would have to infer more accurate description, such as color, speed, and so on.


While a single-letter characterization of the optimal relevance-complexity region of this model can be found easily by specializing the result of \cite[Theorem 1]{tian2007multistage} to the case of logarithmic loss measure, we here establish explicit (analytic) characterizations of the relevance-complexity region for two examples, i.e. memoryless Gaussian and memoryless binary sources. Then, we derive a Blahut-Arimoto type algorithm that allows us to compute (an approximation of) the region for general discrete sources. Finally, we illustrate the 
usefulness of our results through an application example in the pattern classification problem. 


The relevance-complexity tradeoff has been studied in various setups, including the distributed inference \cite{estella2018distributed,aguerri2019distributed}, multivariate \cite{friedman2001multivariate} and multi-layer  \cite{yang2017multi} scenarios. 
Although multivariate IB \cite{friedman2001multivariate} considers a more general setup, the coding method as well as the achievable relevance-complexity region have not been investigated. Notice that our model 
is conceptually different from the multi-layer IB \cite{yang2017multi} that considers non-identical hidden variables (sources) and the input of a given layer as the output of the previous layer.
\begin{figure}[t]
\vspace{-1em}	
\begin{center}	
\includegraphics[width=0.4\textwidth]{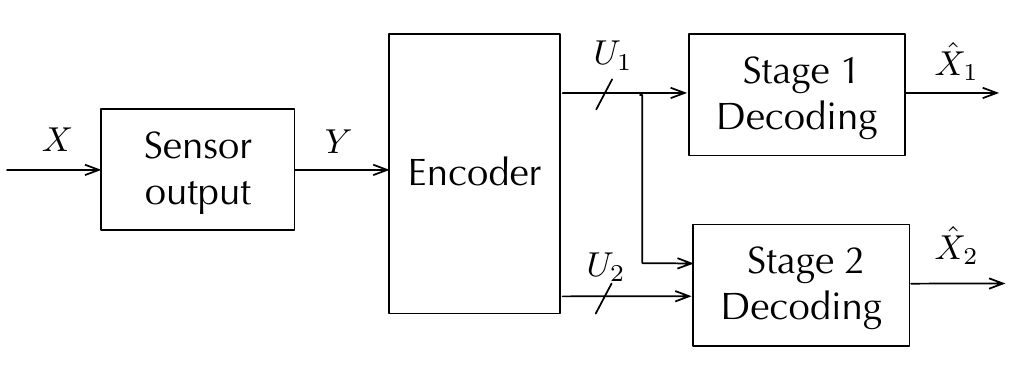}
\caption{Scalable information bottleneck for $T=2$ stages.}
\label{fig:ScalableRSC}
\end{center}
\vspace{-1.2em}	
\end{figure} 	

Throughout this paper, we use the following notation. Upper and lower case letters are used to denote random variables and their realizations, e.g., $X$ and $x$, respectively while calligraphic letters, e.g., $\mathcal{X}$, denote sets. 
$X^n$ denotes $(X_1, X_{2}, \dots, X_n)$ and $[K] =\{1, \dots, K\}$.
We use $\mathcal{P(X)}$ to denote the set discrete probability distributions on $\mathcal{X}$. $H(X)$ denotes the entropy of a discrete random variable $X\in \Xc$ while $h(X)$ denotes the differential entropy of a continuous random variable.  
$h_2(p)\eqdef - p\log_2 p -(1-p) \log_2(1-p)$ for $p\in [0,1/2]$ denotes the binary entropy function. 
\section{Problem Formulation}\label{Problem Formulation}

Let $(X^n, Y^n) \in \Xc^n \times \Yc^n$ be a sequence of $n$ i.i.d. discrete random variables corresponding to the source and observation. 
An $\left(n, R_1, \dots, R_T \right)$ successive refinement 
code for the $T$-scalable information bottleneck consists of
\begin{itemize}
\item Encoding function $\phi_t^{(n)}: \Yc^n \mapsto \Rc_t$ for $t=1, \dots, T$, where $\Rc_t =\left \{1, \dots, 2^{n R_t} \right \}$.
\item Decoding function $\psi_t^{(n)}: \Rc_1\times \dots \times \Rc_t \mapsto \hat{\Xc}_t^{n}$, for $t=1, \dots, T$. 
\end{itemize}
The decoder wishes to reconstruct the source in a scalable fashion, i.e. for $t=1, \dots, T$, 
\begin{align}\label{eq:Distortion}
\EE \left[d^n(X^n, \hat{X}^n_t) \right] =\frac{1}{n}\EE \left[ \sum_{i=1}^n d(X_i, \hat{X}_{t, i})\right]\leq D_t,
\end{align}
where $d(x, \hat{x})$ is a distortion function. In particular, we focus on the logarithmic distortion measure, given by $d_{\rm log}(x, \hat{x}) =  - \log(\hat{x}(x))$, where $\hat{x}(x) \in \Pc(\Xc)$ denotes the probability distribution (mass function) evaluated at $x\in \Xc$. 
We say that the rate-distortion tuples $(R_1, \dots, R_T, D_1, \dots, D_T)$ are achievable if there exists an $(n, R_1, \dots, R_T)$ $T$-scalable successive refinement code satisfying \eqref{eq:Distortion} for an arbitrarily large $n$. 
By letting $\Delta_t$ denote the {\it relevance} of stage-$t$ defined as $\Delta_t \eqdef H(X)-D_t$, the relevance-complexity region is defined as the minimum complexity (rate) tuple $(R_1, \dots, R_T)$ such that the relevance constraints $(\Delta_1, \dots, \Delta_T)$ are satisfied. 

%
%

\section{Relevance-Complexity Region of $T$-Scalable Information Bottleneck}\label{sect:relevance-complexity Region}
This section characterizes the relevance-complexity region of the $T$-scalable information bottleneck problem described in the previous section.
\begin{theorem}\label{theorem:region1}
The relevance-complexity region of $T$-scalable information bottleneck satisfies
\begin{subequations}
\begin{align}
\sum_{l=1}^{t}R_l &\geq \sum_{l=1}^{t}I \left(Y;U_l|U_1, \dots, U^{l-1} \right),\;\;\; \forall t \label{eq:region1-rate} \\
\Delta_t  &  \leq I \left(X;U_1, \dots, U_t \right), \;\;\; \forall t  \label{eq:region1-relevance} 
\end{align}
\label{eq:region1}
\end{subequations}
where $X \mkv Y \mkv (U_1, \dots, U_T)$ forms a Markov chain. 
\end{theorem}

\begin{remark}
\label{Changing Relevance Formulation}
Similarly to \cite[Remark 1]{tian2007multistage}, we can provide an alternative region to Theorem \ref{theorem:region1} under the stronger Markov chain $X \mkv Y \mkv U_T \mkv \dots \mkv U_1$.
\begin{subequations}\label{eq:region2}
\begin{align}
\sum_{l=1}^{t}R_l &\geq \sum_{l=1}^{t} I \left(Y;U_l|U_1, \dots, U_{l-1} \right), \;\;\; \forall t \label{eq:region2-rate} \\
\Delta_t  &  \leq I \left(X;U_t \right), \;\;\; \forall t.  \label{eq:region2-relevance} 
\end{align}
\end{subequations}
These two regions are equivalent as the RHS of \eqref{eq:region1-relevance} coincides with that of \eqref{eq:region2-relevance} under the stronger Markov chain. 
The former \eqref{eq:region1} will be used to characterize the relevance-complexity region in Section \ref{sect:Examples}, while the latter \eqref{eq:region2} will be used to design the algorithm in Section \ref{sect:Algorithm}.
\end{remark} 
The rest of the section is dedicated to the converse proof as the achievability builds on the hierarchical random codes and successive decoding as described in \cite{tian2007multistage}.

\begin{proof}
By recalling the decoder output of stage $t$ as $V_t=\phi_t^{(n)}(Y^n)$ and defining $U'_{t, i} = \left(V_t, Y^{i-1} \right)$ for all $t$ , we have
\begin{align}\label{eq:rate}
n \sum_{l=1}^t R_l & \geq \sum_{l=1}^t H(V_t)  \geq H(V^t) \geq I \left(V^t; Y^n \right)\nonumber\\
&\stackrel{(a)}
= \sum_{i=1}^n \sum_{l=1}^t I \left( V_l; Y_i |Y^{i-1}, V^{l-1} \right) \nonumber\\
&\stackrel{(b)}=  \sum_{i=1}^n \sum_{l=1}^t I \left( V_l, Y^{i-1}; Y_i |Y^{i-1}, V^{l-1} \right) \nonumber\\
&= \sum_{i=1}^n \sum_{l=1}^t I \left( U'_{l, i}; Y_i | U'_{1, i} , \dots, U'_{l-1, i} \right)\nonumber \\
& \stackrel{(c)}= n \sum_{l=1}^t I \left( U'_{l, Q}; Y_Q | U'_{1, Q} , \dots, U'_{l-1, Q}, Q \right)\nonumber \\
& \stackrel{(d)}= n \sum_{l=1}^t I \left( U_{l}; Y | U_{1} , \dots, U_t \right)
 \end{align}
 where (a) follows from the chain rule; (b) follows from $I \left( Y^{i-1}; Y_i |Y^{i-1}, V^{l} \right)=0$; (c) follows by definining a uniformly distributed random variable $Q$ over $[1:n]$ independent of all other variables; (d) follows by identifying $Y=Y_Q, U_{t} = (U'_t, Q)$.

Next, we derive the lower bound on the distortion. 
By applying \cite[Lemma 1]{courtade2013multiterminal} to each stage $t$,
we have
 \begin{align}\label{eq:distortion}
 nD_t & \geq H(X^n | V_1, \dots, V_t) \nonumber\\
 &\stackrel{(a)} = \sum_{i=1}^n H \left(X_i | X^{i-1}, V_1, \dots, V_t \right)\nonumber\\
 &\stackrel{(b)} \geq \sum_{i=1}^n H \left(X_i | X^{i-1}, Y^{i-1}, V_1, \dots, V_t \right) \nonumber\\
 &\stackrel{(c)} = \sum_{i=1}^n H \left(X_i |U'_{1,i}, \dots, U'_{t,i} \right) \nonumber\\
&\stackrel{(d)} = n H \left(X| U_1, \dots, U_t \right) 
 \end{align}
 where (a) follows from the chain rule; (b) follows because conditioning decreases the entropy; (c)
 follows from the Markov chain $X_i \mkv \left(U'_{1,i}, \dots, U'_{t,i} \right) \mkv X^{i-1}$; in (d) we used $X= X_Q$ and $U_t= (U'_t, Q)$.  By substituting $\Delta_t = H(X)-D_t$, we readily obtain the desired region.
\end{proof}

%

\section{Examples}\label{sect:Examples}
\subsection{Binary Example}\label{subsect:Binary Example}
We assume that the source $X$ and the observation $Y$ are both binary and Bernoulli distributed with $\frac{1}{2}$, denoted by 
$\Bern(1/2)$, and 
they are related through the binary symmetric channel $X=Y\oplus N$, where $N\sim \Bern(p)$ is independent of $X$ with $p\in [0, 1/2]$. 
\begin{proposition}\label{prop:binary_proposition}
The relevance-complexity region of the scalable IB for the binary memoryless source and observation is given by 
\begin{align}\label{the other binary region}
\Delta_t \leq 1-h_2 \left(h_2^{-1}\left(1-\sum_{l=1}^tR_l \right) \star p \right), \;\; \forall t
\end{align}
where $\{\Delta_t\}$ satisfies 
\begin{align}
\label{condition_the other binary region}
0 \leq \Delta_1 \leq \dots \leq \Delta_T \leq 1-h_2(p).
\end{align}
\end{proposition}

\begin{proof}
Our proof consists of two steps, namely, we first prove the intermediate relevance-complexity region $\Rc_B$ in \eqref{binary_region_both}, and then prove the equivalence between $\Rc_B$ and the region in Proposition \ref{prop:binary_proposition}, denoted by $\Rc'_B$. 
\begin{subequations}\label{binary_region_both}
\begin{align}
&\sum_{l=1}^tR_l \geq 1-h_2(\alpha_t),\;\;\; \forall t\label{binary_complexity} \\
&\Delta_t \leq 1-h_2(\alpha_t \star p),\;\;\; \forall t\label{binary_relevance}
\end{align}
\end{subequations}
where $\{\alpha_t\}$ satisfies 
\begin{align}\label{eq:alpha}
0 \leq \alpha_T \leq \dots \leq \alpha_t \leq \dots \leq \alpha_1 \leq \frac{1}{2}.
\end{align} 
{\bf Step 1:~} We rewrite the relevance-complexity region \eqref{eq:region1} as
\begin{subequations}\label{eq:562}
\begin{align}
\sum_{l=1}^{t}R_l  &\geq  H(Y)-H(Y| U^t),\;\;\; \forall t \\ \label{eq:563}
\Delta_t  &  \leq H(X) - H(X|U^t), \;\;\; \forall t
\end{align}
\end{subequations}
First we look at the relevance term in \eqref{eq:563}. The Markov chain $X \mkv Y \mkv U^T$ implies 
\begin{align}
H(X|Y) \leq  H(X|U^t) \leq H(X). 
\end{align}
Due to the continuity of the entropy, there exists a unique value $\alpha_t\in [0,1/2]$ satisfying
\begin{align}\label{eq:571}
H(X|U^t) = h_2( p \star \alpha_t)
\end{align}
which is larger than $H(X|Y)=h_2(p)$ corresponding to $\alpha_t=0$ and
smaller than $H(X)=1$ corresponding to $\alpha_t=1/2$. Since $H \left(X|U^T \right) \leq \dots\leq H(X|U_1)$, we readily obtain \eqref{eq:alpha}.
Next, we wish to lower bound the complexity term by finding the upper bound of the term $H(Y| U^t)$. By writing $X= Y \oplus Z$ with $Z\sim\Bern(p)$ such 
that $Z$ is independent of $(Y, U^T)$ \footnote{Alternatively we can write $Y= X+Z$. However, this will not
satisfy the condition that $Z$ is independent of $(X, U_1, \dots, U_T)$.}, 
we can apply Mrs. Gerber's lemma \cite{el2011network}
\begin{align}\label{eq:580}
H(X |U^t) \geq h_2 \left( h_2^{-1} (H(Y|U^t) ) \star p \right).
\end{align}
By combining \eqref{eq:571} and \eqref{eq:580}, and further noticing that $h_2^{-1}(x)$ is increasing in $x\in [0,1]$, we obtain 
\begin{align}\label{eq:584}
H(Y| U^t) \leq h_2(\alpha_t) 
\end{align}
Plugging \eqref{eq:571} and \eqref{eq:584} into \eqref{eq:563} and \eqref{eq:562} respectively, we obtain \eqref{binary_complexity} and \eqref{binary_relevance}. 
This establishes the converse proof for the region \eqref{binary_region_both}

For the achievability, we let 
\begin{align} \label{eq:cascade}
U_t &=  Y \oplus Z_t  \;\;  Z_{t}\sim\Bern(\alpha_t), \forall t 
\end{align}
or alternatively write $U_t = U_{t+1} \oplus \tilde{Z}_{t} $ with $U_t = U_{t+1} \oplus \tilde{Z}_{t}$ for $t\leq T-1$. Plugging \eqref{eq:cascade} into the RHS of \eqref{eq:562}, we obtain the desired region in \eqref{binary_region_both}. 

{\bf Step 2:} Now we prove that $\mathcal{R}_B$ in \eqref{binary_region_both} is equivalent to the region $\mathcal{R}'_B$ of Proposition \ref{the other binary region}.
To this end, we verify both $\mathcal{R}_{B}\subseteq \mathcal{R'}_{B}$ and $\mathcal{R'}_{B}\subseteq \mathcal{R}_{B}$. In the former case, we show that every point $(\bm{\delta},\bm{r}) \in \mathcal{R}_{B}$, where $\bm{\delta}=[\delta_1, \dots, \delta_T]$ and $\bm{r}=[r_1, \dots, r_T]$, is also a point in the region $\mathcal{R'}_{B}$. For every pair of $(\delta_t,r_t)$, from \eqref{binary_complexity} and noticing that $h_2(p)$ is monotonically increasing in $p\in [0,1/2]$, we readily obtain 
\begin{align}
\label{eq:501}
\alpha_t &\geq h_2^{-1} \left(1-\sum_{l=1}^tr_l \right), 
\end{align}
Plugging \eqref{eq:501} into \eqref{binary_relevance}, we obtain the region $\mathcal{R'}_B$ as follows
\begin{align}
\delta_t \leq 1-h_2 \left(h_2^{-1} \left(1-\sum_{l=1}^tr_l \right) \star p \right),
\end{align}
for
\begin{align}\label{condition_the other binary region}
0 \leq \delta_1 \leq \dots \leq \delta_T \leq 1-h_2(p),
\end{align}
where \eqref{condition_the other binary region} is obtained by combining \eqref{binary_relevance} and \eqref{eq:alpha}
and recalling that the binary entropy function $h_2(\alpha_t \star p)$ is monotonically increasing in $\alpha_t$ for a fixed $p$.

Next, we prove that every point $(\bm{\delta},\bm{r}) \in \mathcal{R}'_{B}$ is also a point in the region $\mathcal{R}_{B}$. For all pairs of $(\delta_t,r_t)$, from \eqref{the other binary region} we have
\begin{align}
\delta_t \leq 1-h_2 \left(\alpha'_t \star p \right),
\end{align}
with 
\begin{align}\label{eq:915}
\alpha'_t = h_2^{-1} \left(1-\sum_{l=1}^t r_l \right) 
\end{align}
From \eqref{eq:915}, we readily obtain $0\leq \alpha'_T \leq \dots \leq \alpha'_1\leq 1/2$. 
Thus, each point $(\delta_t,r_t)\in \mathcal{R'}_{B}$ is also in the region of $\mathcal{R}_{B}$  and the proof is complete.
\end{proof}
\begin{figure*}[ht]
     \centering
     \subfloat[ ][Relevance-complexity tradeoff $(R,\Delta_1)$ with $\Delta_2=0.11$.]{
     \includegraphics[width=0.45\textwidth]{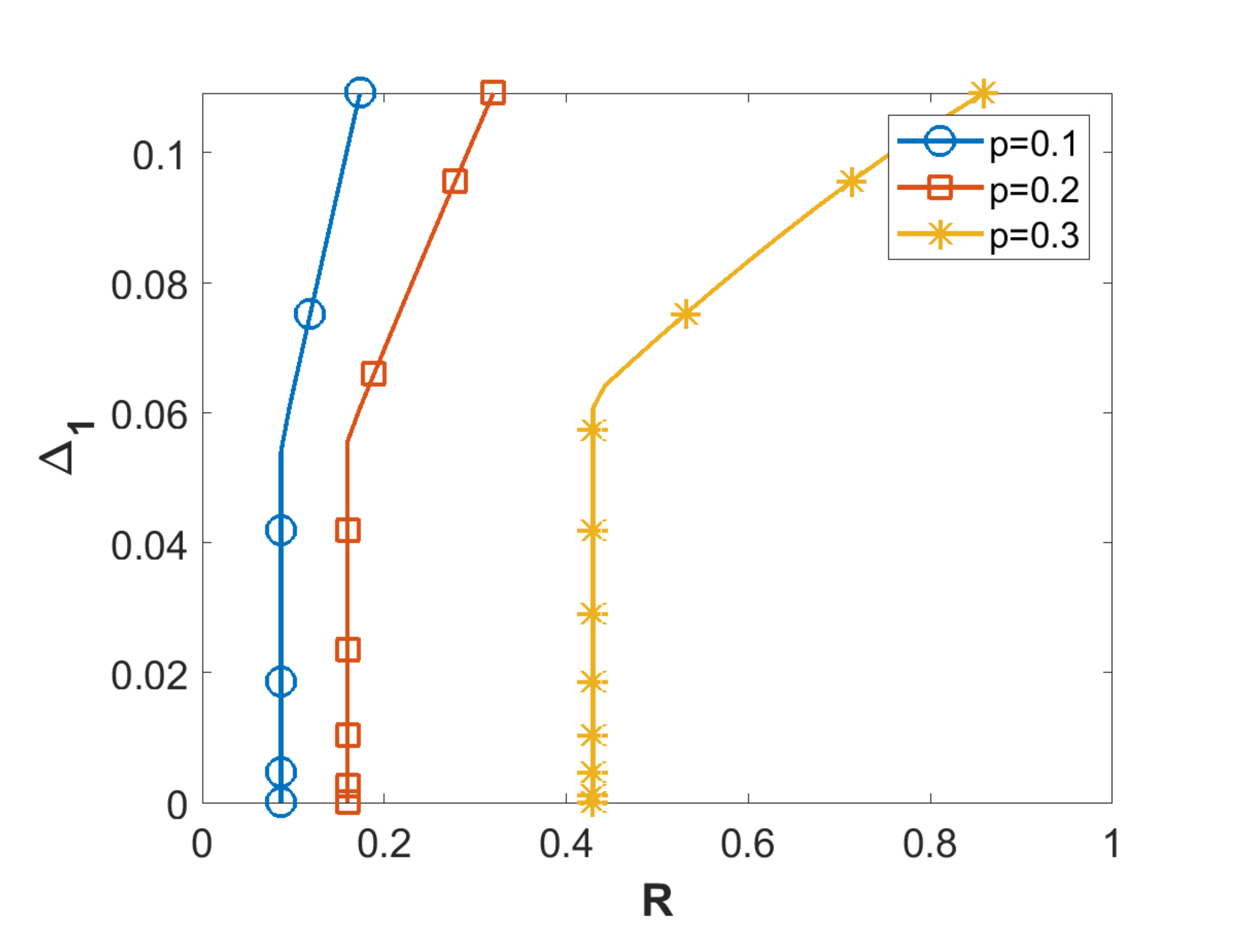}
     \label{fig:BR_11}
     }
     \subfloat[ ][Relevance-complexity tradeoff $(R,\Delta_2)$ with $\Delta_1=0$.]{
     \includegraphics[width=0.45\textwidth]{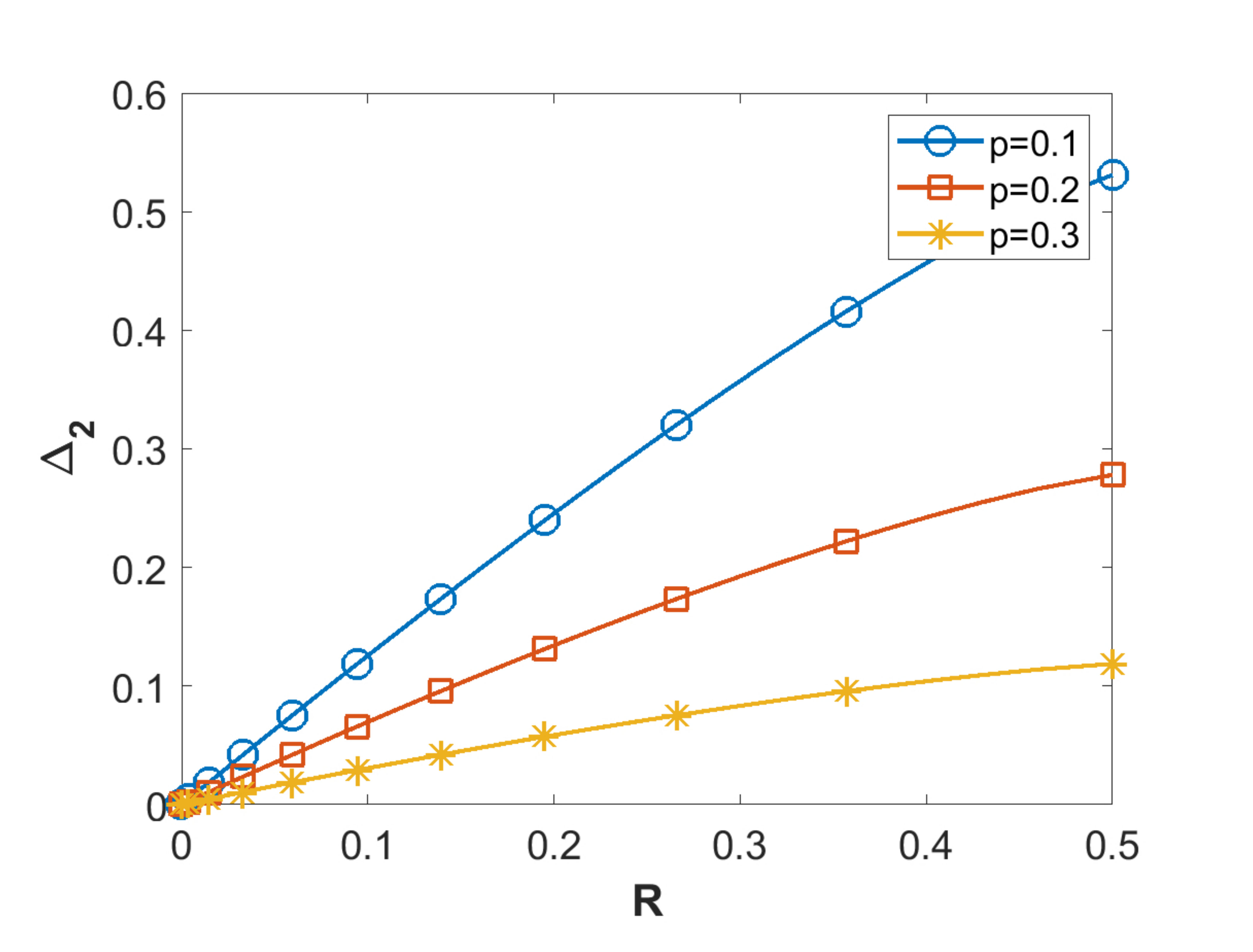}
     \label{fig:BR_12}
     }
     \caption{Comparing the regions $(R,\Delta_1)$ and $(R,\Delta_2)$ for different crossover probabilities $p$.}
     \label{fig:BR_1112}
\end{figure*}
\vspace{-1em}

\subsection{Scalar Gaussian Example}\label{subsect:Scalar Gaussian Example}
We consider the scalar Gaussian source and observation whose relation is given by 
\begin{align}
Y = X + W
\end{align}
where $X, W$ are independent real Gaussian random variables with zero mean and variance $\sigma_x^2, \sigma_w^2$ respectively. 
\begin{proposition}\label{prop:Gaussian}
The relevance-complexity region $\mathcal{R}_{SG}$ of the scalable IB for the scalar Gaussian source and observation is given by 
\begin{align}\label{the other SG region}
\Delta_t \leq \frac{1}{2}\log \left(\frac{\sigma^2_x+\sigma^2_w}{2^{\left(-2\sum_{l=1}^t R_l \right)}\sigma^2_x+\sigma^2_w} \right),\;\; \forall t 
\end{align}
where 
\begin{align}
\label{condition_the other SG region}
0 \leq \Delta_1 \leq \dots \leq \Delta_T \leq \frac{1}{2}\log \left( \frac{\sigma_x^2+\sigma_w^2}{\sigma_w^2} \right).
\end{align}
\end{proposition}
\begin{proof} 
Similarly to the binary example, we prove Proposition \ref{prop:Gaussian} in two steps. First we prove the intermediate
relevance-complexity region $\mathcal{R}_{SG}$ given by \eqref{Scalar_Gaussian_equations}. Then, we prove the equivalence between $\mathcal{R}_{SG}$ and the region $\mathcal{R}'_{SG}$ in Proposition \ref{prop:Gaussian}.
\begin{subequations}\label{Scalar_Gaussian_equations}
\begin{align}
\sum_{l=1}^{t}R_l  &\geq  \frac{1}{2} \log  \left(\frac{\gamma_t+\sigma^2_x}{\gamma_t-\sigma^2_w} \right),\;\; \forall t \label{eq:Grate} \\
\Delta_t  &  \leq \frac{1}{2} \log \left(1+\frac{\sigma^2_x}{\gamma_t} \right),  \;\; \forall t \label{eq:Grelevance}
\end{align}
\end{subequations}
where $\{\gamma_t\}$ satisfies
\begin{align}\label{eq:Gorder}
\sigma^2_w \leq \gamma_T \leq \dots \leq \gamma_1 .
\end{align}

We first provide the converse proof.  In order to find the upper bound of \eqref{eq:563}, we lower bound the term $h(X|U^t)$. 
Following similar steps as the binary example, we have
\begin{align}\label{eq:615}
h(X|Y) \leq  h(X|U^t) \leq h(X), \;\; \forall t.
\end{align}
Since the optimal MMSE estimate of a Gaussian source $X$ given a noisy observation $Y$ is $\hat{X}=\frac{\sigma^2_x }{\sigma^2_x+\sigma^2_w} Y$, we can write $ X= \hat{X} + \tilde{X}$, where the estimate $\hat{X}$ is independent of the estimation error $\tilde{X}$ and also $\hat{X}$ is Gaussian with zero mean and variance $\sigma_e^2= (1/\sigma^2_x+1/\sigma^2_w )^{-1} $, yielding $h(X|Y) = \frac{1}{2}  \log\left( 2\pi e \sigma_e^2 \right)$. 
Hence, the inequalities \eqref{eq:615} imply that there exists a value $\gamma_t$ satisfying 
\begin{align}\label{eq:619}
h(X|U^t) = \frac{1}{2} \log\left( 2\pi e \left(\frac{1}{\sigma^2_x}+\frac{1}{\gamma_t} \right)^{-1} \right)
\end{align}
for $\gamma_t\in [\sigma^2_w, \infty)$. Combining \eqref{eq:619} and $h(X)=\frac{1}{2}\log(2\pi e \sigma_x^2)$, and using $h(X|U^T) \leq \dots \leq h(X|U_1)$, we readily obtain \eqref{eq:Grelevance} with \eqref{eq:Gorder}. 

Next, we wish to lower bound \eqref{eq:563} by finding the upper bound of $h(Y| U^t)$. 
Since $\hat{X}$ and $\tilde{X}$ are conditionally independent given any function of $Y$, 
 $\hat{X}$ and $\tilde{X}$ are also conditionally independent given $U_1, \dots, U_T$. Therefore, we can apply the conditional EPI \cite{el2011network} as
 \begin{align}\label{eq:644}
 h(X|U^t) & = h(\hat{X} + \tilde{X} |U^t) \nonumber\\
 &\geq \frac{1}{2} \log\left( 2^{2h(\hat{X}|U^t) } + 2^{2h(\tilde{X} |U^t)}\right) \nonumber \\
 &\stackrel{(a)}= \frac{1}{2} \log\left( 2^{2( \log a + h(Y|U^t)) } + 2^{2h(\tilde{X} |U^t)}\right)\nonumber \\
  &\stackrel{(b)}\geq \frac{1}{2} \log\left( 2^{2 ( \log a + h(Y|U^t) )} + 2^{2h(\tilde{X} |Y,U^t)}\right)\nonumber \\
  &\stackrel{(c)}\geq \frac{1}{2} \log\left( 2^{2(\log a + h(Y|U^t)) } + 2^{2h(\tilde{X} |Y)}\right)\nonumber \\
 &\stackrel{(d)}= \frac{1}{2} \log\left(\! 2^{2 (\log a \!+\! h(Y|U^t)) } + 2 \pi e \sigma_e^2 \! \right)
 \end{align} 
 where (a) follows because $\hat{X}=a  Y$ and $h(a X) = h(X) + \log |a|$ for $a=\frac{\sigma^2_x }{\sigma^2_x+\sigma^2_w}$; (b) follows because the conditioning decreases the entropy; (c) follows because $U_1, \dots, U_t$ are 
 functions of $Y$; (d) follows because $\tilde{X}$ and $Y$ are independent. 
 
Combining \eqref{eq:619} and \eqref{eq:644}, we obtain
\begin{align}\label{eq:661}
h(Y|U^t) \leq \frac{1}{2} \log\left[2\pi e \frac{(\gamma_t-\sigma^2_w)(\sigma^2_x+\sigma^2_w)}{\gamma_t+\sigma^2_x} \right]
\end{align}
Combining \eqref{eq:661} and $h(Y)=\frac{1}{2}\log 2\pi e(\sigma_x^2+\sigma_w^2)$, we obtain \eqref{eq:Grate}. This completes the converse proof.

For the achievability, we let $U_t=Y+Z_t$ such that $Z_{t}\sim\Nc (0, \gamma_t-\sigma^2_w)$. By evaluating $h(Y|U^t)$ and $h(X|U^t)$, we can easily show that 
this choice achieves the region \eqref{Scalar_Gaussian_equations}.

The proof of the equivalence between $\mathcal{R}_{SG}$ and the region $\mathcal{R}'_{SG}$ in Proposition \ref{prop:Gaussian} follows the same footsteps as Step 2 of Proposition \ref{prop:binary_proposition}, hence is omitted. By letting 
\begin{align}\label{gamma_t scalar Gaussian}
\gamma_t=\frac{2^{\left(-2\sum_{l=1}^t r_l \right)}\sigma^2_x+\sigma^2_w}{1-2^ {\left(-2\sum_{l=1}^t r_l \right)}}.
\end{align}
in \eqref{Scalar_Gaussian_equations}, we obtain the desired expression \eqref{the other SG region}. 


\end{proof}
\begin{figure*}[ht]
     \centering
     \subfloat[ ][Relevance-complexity tradeoff $(R,\Delta_1)$ with $\Delta_2=0.5$.]{
     \includegraphics[width=0.45\textwidth]{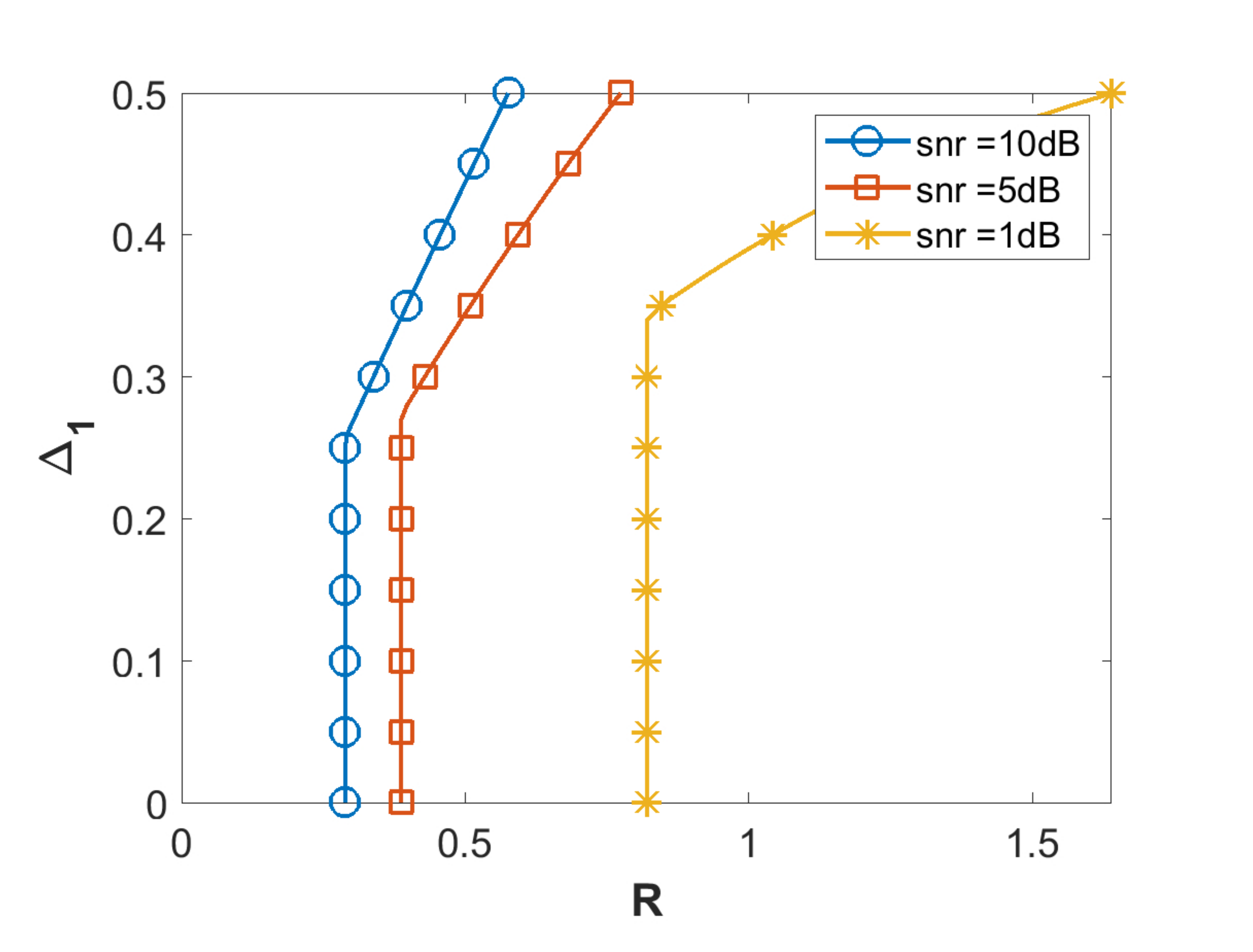}
     \label{fig:GR_07_b}
     }
     \subfloat[ ][Relevance-complexity tradeoff $(R,\Delta_2)$ with $\Delta_1=0$.]{
     \includegraphics[width=0.45\textwidth]{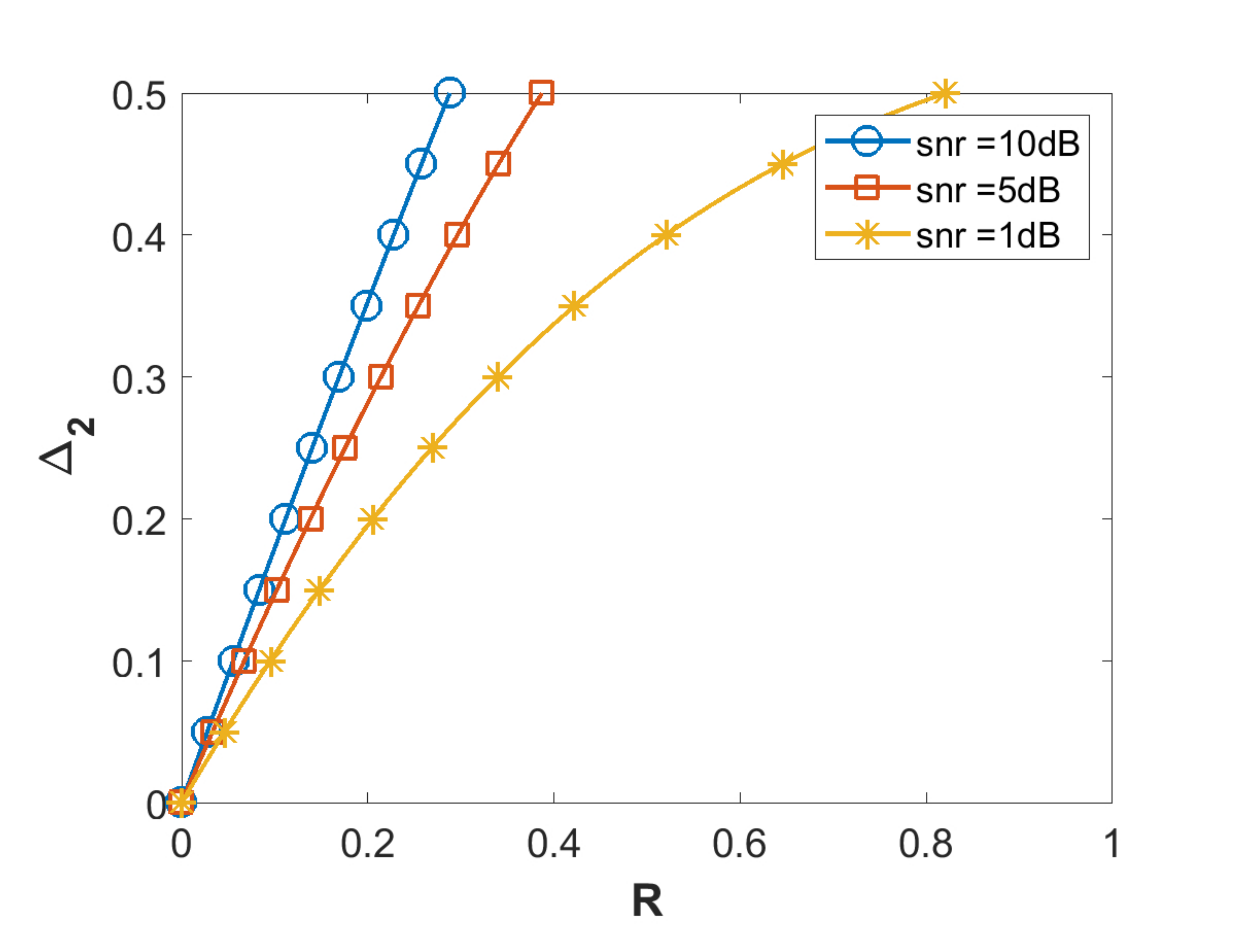}
     \label{fig:GR_08_b}
     }
     \caption{Comparing the regions $(R,\Delta_1)$ and $(R,\Delta_2)$ for different values of $\snr$.}
     \label{fig:GR_78_b}
\end{figure*}
\vspace{-1em}
\subsection{Numerical Evaluation}
We evaluate the relevance-complexity region of the binary and Gaussian examples in Propositions \ref{prop:binary_proposition} and \ref{prop:Gaussian}, for the case of two stages $T=2$ and focusing on the symmetric rate $R=R_1=R_2$. The tradeoff between the symmetric rate $R$ and the relevance $\Delta_1, \Delta_2$ for the binary case is 
\begin{align}\label{eq:SymR}
&R \geq \max\left\{ 1-h_2 \left( \frac{h_2^{-1}(1-\Delta_1)-p}{1-2p} \right),  \right. \nonumber\\
&\qquad\qquad\qquad \left. {} \frac{1}{2}-\frac{1}{2}h_2 \left( \frac{h_2^{-1}(1-\Delta_2)-p}{1-2p} \right) \right\}.
\end{align}
The same tradeoff for the Gaussian case is given by 
\begin{align}\label{eq:SymR_Gaussian}
&R \geq \text{max}\left\{ \frac{1}{2}\log \left(\frac{\snr}{(1+\snr) 2^{(-2\Delta_1)}-1}\right),\right.\nonumber\\
&\qquad \left. {} \frac{1}{4}\log \left(\frac{\snr}{(1+\snr) 2^{(-2\Delta_2)}-1}\right) \right\}.
\end{align}
where we let $\snr=\sigma_x^2/\sigma_n^2$. 
Fig. \ref{fig:BR_1112} shows the tradeoff for different values of crossover probability $p$. We observe that when the observation is less noisy with smaller value of $p$, the relevance $\Delta_t$ increases for any complexity $R$, for $t=1,2$. In Fig. \ref{fig:BR_11}, the relevance-rate tradeoff has a 
threshold point. For $p=0.2$ (red curve), by equalizing two terms inside the maximum in \eqref{eq:SymR}, we obtain the threshold point 
$(\Delta_1, R)=  (0.056,0.16)$. Namely, we can achieve any smaller value of $\Delta_1$ than $0.056$ for the same complexity $R=0.16$ while 
a larger relevance can be achieved only by increasing the complexity. 

Fig. \ref{fig:GR_78_b} shows the same tradeoff for different $\snr$ values. Similar to the binary case, in Fig. \ref{fig:GR_07_b} there are threshold points in the tradeoff. The threshold point for $\snr=5dB$ is $(0.27,0.39)$, which means we can achieve any smaller value of $\Delta_1$ than $0.27$ for the same complexity $R=0.39$ while a larger relevance can be achieved only by increasing the complexity.

\section{Modified Blahut-Arimoto Algorithm}\label{sect:Algorithm}
In this Section, we propose a modified Blahut-Arimoto (BA) (see e.g. \cite{blahut1972computation} algorithm which computes the region by iterating over a set of self-consistent equations when the source distribution is known or can be estimated with high accuracy. 
From the relevance-complexity region of \eqref{eq:region2}, we wish to minimize 
the sum complexity under the individual relevance constraints. 
\begin{align}
\label{optimization problem}
\text{min}_{p(u_1, \dots, u_T|y)}\;\; I(U_1,\dots, U_T;Y)-\sum_{l=1}^T\beta_l I(U_l;X).
\end{align}
We define the corresponding Lagrangian function as
\begin{align}
\label{eq_algorithm_02}
\mathcal{L}&= I(U^T;Y)-\sum_{l=1}^T\beta_l I(U_l;X)-\sum_{y}\lambda(y)\sum_{u^{T}}p(u^{T}|y) \nonumber\\
& = \sum_{y, u^{T}}p(y)p(u^{T}|y)\log \left(\frac{p(u^{T}|y)}{p(u^{T})} \right)-\sum_{y}\lambda(y)\sum_{u^{T}}p(u^{T}|y),\nonumber\\
&-\sum_{l=1}^{T}\beta_{l}\sum_{x, u_{l}}p(x)p(u_{l}|x)\log \left(\frac{p(u_{l}|x)}{p(u_{l})} \right)
\end{align}
where $\lambda (y)$ is the normalization term. By simple algebra, we readily obtain
\begin{align}
\label{eq_algorithm_03}
p(u^{T}|y)=\frac{p(u^{T})}{Z(y)}\prod_{l=1}^{T} \prod_{x}2^{\eta(l,x)-\gamma(l,x) },
\end{align}
where 
\begin{align*}
&\frac{1}{Z(y)}=2^{\lambda(y)}, \;\;\;  \gamma(l,x) = \beta_l p(u_l,x) / p(u^T), \\
&\eta(l,x) =\left[ \beta_l p(x|y) \left(1+\log \left(\frac{p(u_l|x)}{p(u_l)} \right) \right) \right]/\left[p(u_{\Tc_{\overline{l}}}|x,u_l) \right],
\end{align*}
and $\Tc_{\overline{l}}= [T]\setminus \{l\}$. 
Finally, our modified Blahut-Arimoto algorithm is summarized in Algorithm \ref{alg:BA_discrete}.
 \begin{algorithm}
 \caption{Blahut-Arimoto type Algorithm for known discrete sources}
 \label{alg:BA_discrete}
 \begin{algorithmic}[1]
 \renewcommand{\algorithmicrequire}{\textbf{Input:}}
 \renewcommand{\algorithmicensure}{\textbf{Output:}}
 \REQUIRE pmf $P_{X,Y}$, parameters $\beta_l, l \in [T]$.
 \ENSURE  $\{P_{U_t|Y}^*\}$ and $(\{R_t\}, \{\Delta_t\})$.
 \\ \textit{Initialisation} : Set $k=0$. Set $p^{(0)}(u^T|y)$ randomly.
  \REPEAT
   	\STATE Update the following pmfs for $t=1, \dots, T$
  	\STATE $p^{(k+1)}(u_t)=\sum_{u_{\Tc_{\overline{t}}}} \sum_{y}p^{(k)}(u^T|y)p(y)$,
  	\STATE $p^{(k+1)}(u_t|x)=\sum_{u_{\Tc_{\overline{t}}}}\sum_{y}p(y|x)p^{(k)}(u^T|y)$,
  	\STATE $p^{(k+1)}(u_{\Tc_{\overline{t}}}|u_t,x)=\frac{\sum_{y}p(y|x)p^{(k)}(u^T|y)}{p^{(k+1)}(u_t|x)}$,
  	\STATE Update $p^{(k+1)}(u^T|y)$ by using \eqref{eq_algorithm_03}.
  	\STATE $k \leftarrow k+1$
  \UNTIL {convergence.}
 \end{algorithmic} 
 \end{algorithm}

\section{Application to Pattern Classification}\label{sect:Classification}
We conclude the paper by providing an application of the proposed scalable IB in the context of pattern classification. 
Consider the problem of classification in the proposed stucture shown in Fig. \ref{fig:ScalableRSC}. 
In this problem, the role of the decoder is to guess the unknown class $X \in \mathcal{X}^C$, for the number of $C$ classes, on the basis of encoder outputs $U_1, \dots ,U_T$. Let $Q_{U^T|Y}$ be the encoder and $Q_{\hat{X}_t|U^t}$ be the soft-decoder output of stage $t \in [T]$. The pair of encoder and decoder induces a classifier at stage $t\in [T]$
\begin{align}\label{eq:Class_Error_01}
Q_{\hat{X}_t|Y} &=\sum_{u_1,\dots, u_t}  Q_{\hat{X}_t,U^t|Y}  \stackrel{(a)}=\sum_{u_1,\dots, u_t}  Q_{U^t|Y}Q_{\hat{X}_t|U^t}  \nonumber\\
&= \EE_{Q_{U^t|Y}} [Q_{\hat{X}_t|U^t} ],
\end{align}
where (a) follows because we have $Y \mkv U^t \mkv \hat{X}_t$. 
Following similar steps as \cite[Section IV. C]{ugur2020vector}, we can show that 
the probability of classification error at stage $t$ of the $T$-scalable classiffer is given by 
\begin{align}\label{eq:Class_Error_02}
&P_{\epsilon,t}^{(T)}(Q_{\hat{X}^t|Y})  = 1 - \EE_{P_{X,Y}}[Q_{\hat{X}^t|Y}]\nonumber\\
& \stackrel{(a)} \leq 1 - 2^{-\EE_{P_{X,Y}}[-\log(Q_{\hat{X}^t|Y})]}\nonumber\\
&\stackrel{(b)} = 1 - 2^{-\EE_{P_{X,Y}}\left[-\log \left(\EE_{Q_{U^t|Y}} \left[ Q_{\hat{X}_t|U^t}  \right] \right)\right]}\nonumber\\
&\stackrel{(c)} \leq 1 - 2^{-\EE_{P_{X,Y}}\EE_{Q_{U^T|Y}} \left[-\log \left( Q_{\hat{X}_t|U^t}  \right)\right]}
\end{align}
where (a) and (c) follows from Jensen's inequality; (b) follows from \eqref{eq:Class_Error_01}. By  minimizing the average logarithmic loss in \eqref{eq:Class_Error_02}, we obtain a tight lower bound on classification error probability. 
The minimum average logarithmic loss over all choices of decoders $Q_{\hat{X}_t|U^t}$ for any $t \in [T]$ and the encoder $Q_{U^T|Y}$ with rate no greater than $R_t$ bits per-sample in stage $t$ respectively, is $D_t^*=\inf\{D_t:(R_1, \dots, R_T,D_1, \dots, D_T) \in \mathcal{R} \}$, where $\mathcal{R}$ is the region stated in Theorem \ref{theorem:region1} and $t \in [T]$. Thus, we have
\begin{align}\label{eq:Class_Error_03}
&P_{\epsilon,t}^{(T)}(Q_{\hat{X}^t|Y})  \leq 1 - 2^{- D_t^*}= 1 - 2^{\Delta_t^*-H(X)},
\end{align}
where the last equality follows from $\Delta_t^* = H(X) - D_t^*$, $t \in [T]$. 

In Fig. \ref{fig:Class_Error_01}, we evaluate the upper bound of the binary classification error probability for the case of $T=3$ and symmetric rate $R_1=R_2=R_3=R$, by combining \eqref{the other binary region} and \eqref{eq:Class_Error_03}.
We can see that as the observation gets closer to the source (small $p$) the error probability smoothly decreases. 

\begin{figure}[ht]
  \centering
    \includegraphics[width=1\columnwidth]{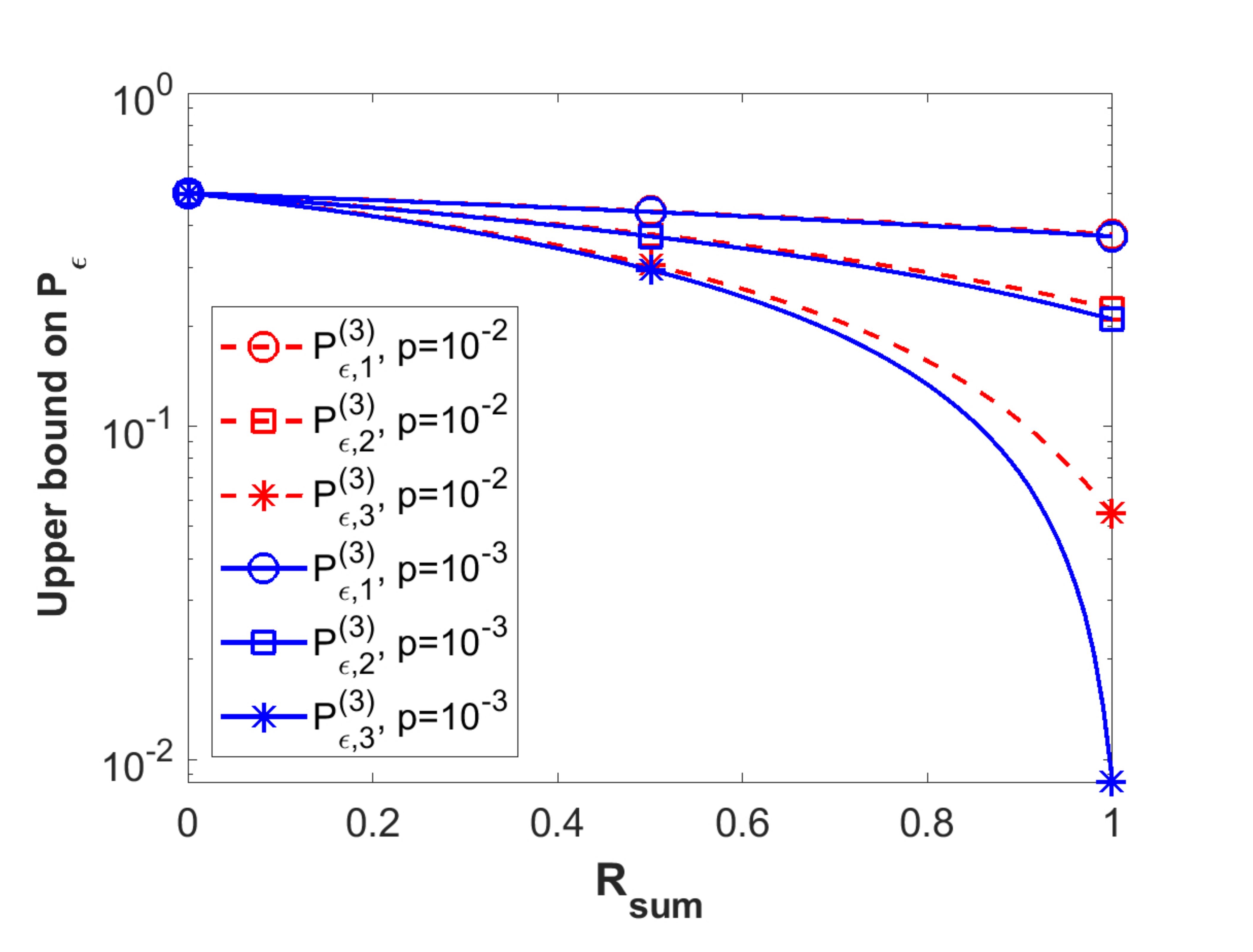}
  \caption{The upper bound of the binary classification error probability for the case of $T=3$ and symmetric rate $R_{sum}=3R$. Classification error is shown in logarithmic scale.}
  \label{fig:Class_Error_01}
\end{figure}

\bibliographystyle{IEEEbib}
\bibliography{References}

\end{document}